\documentclass[pdflatex,sn-basic,icol]{sn-jnl}

\usepackage{graphicx}%
\usepackage{multirow}%
\usepackage{amsmath,amssymb,amsfonts}%
\usepackage{amsthm}%
\usepackage{mathrsfs}%
\usepackage[title]{appendix}%
\usepackage{xcolor}%
\usepackage{textcomp}%
\usepackage{manyfoot}%
\usepackage{booktabs}%
\usepackage{algorithm}%
\usepackage{algpseudocode}%
\usepackage{listings}%
\usepackage{url}%

\renewcommand{\abstractheadfont}{\normalfont\fontsize{9bp}{11bp}\selectfont\titraggedcenter}

\newcommand{\dd}{\mathrm{d}}
\newcommand{\E}{\mathbb{E}}
\newcommand{\Prob}{\mathbb{P}}
\newcommand{\ind}{\mathbf{1}}
\newcommand{\Exp}{\mathrm{Exp}}
\newcommand{\Beta}{\mathrm{Beta}}

\theoremstyle{plain}
\newtheorem{proposition}{Proposition}

\theoremstyle{definition}

\newtheorem{empiricalclaim}{Claim}

\newcommand{\placeholder}[1]{\fbox{\begin{minipage}[c][0.18\textheight][c]{0.94\linewidth}\centering #1\end{minipage}}}

\begin{document}

\title[Phantom-Conditioned Nested Sampling]{Phantom-Conditioned Nested Sampling}

\author*[1,2]{\fnm{Joshua G.} \sur{Albert}}\email{jgalbert@caltech.edu}
% Use a full stop after the sole corresponding author's email address.
\ExplSyntaxOn
\tl_replace_once:Nnn \corrauthemail { ; } { . }
\ExplSyntaxOff

\affil[1]{\orgdiv{Department of Astronomy}, \orgname{California Institute of Technology}, \orgaddress{\postcode{CA 91125}, \city{Pasadena},  \country{USA}}}
\affil[2]{\orgdiv{Leiden Observatory}, \orgname{Leiden University}, \orgaddress{\street{PO Box 9513}, \city{Leiden}, \postcode{2300}, \country{The Netherlands}}}

\abstract{
\mdseries\unboldmath
Nested sampling estimates the evidence by assigning prior volumes to an ordered sequence of likelihood contours.
Markov chain constrained-prior samplers, which are used in high-dimensional problems, generate many intermediate states before producing the next classic sample.
These intermediate states are discarded because their correlation prevents them from being inserted into the ordered NS sequence without changing its order-statistic law.
This paper introduces a novel method of using them to improve evidence estimation, by formulating NS in a Bayesian way and conditioning on phantom samples as Monte Carlo observations.
We then present the open-source software package, JAXNS v3, and its implementation choices.
We validate the approach on a set of problems, and identify its limitations via ablation.
In our experiments, when problem structure is well resolved, conditioning on all retained phantom samples reduces log-evidence RMSE by at least $30\%$, with larger improvements at higher dimensionality.
For the tested problems with unresolved structure, full phantom conditioning produces no detectable improvement or deterioration in evidence accuracy.
Phantom conditioning produces overconfident evidence uncertainties.
We also introduce two dynamic nested sampling allocation schemes.
Evidence-improving allocation approximately halves the number of likelihood evaluations required to achieve comparable evidence accuracy relative to uniform allocation.
Posterior-improving allocation doubles the classic posterior's Kish effective sample size for $17.5\%$ additional likelihood evaluations.
}


\maketitle

\section{Introduction}\label{sec:introduction}

Nested sampling (NS) \citep{2004AIPC..735..395S} rewrites Bayesian evidence estimation as a one-dimensional Lebesgue integration problem.
For standard background, terminology, and a broad review of modern implementations we refer to \citet{2023StSur..17..169B}.
Here we only introduce the notation needed for the present construction.
Let $(\mathcal{A},\Sigma,p)$ be a probability space and let $L:\mathcal{A}\to[0,\infty)$ be a measurable likelihood.
The evidence is
\begin{align}
    Z = \int_\mathcal{A} L(x)\,\dd p(x).\label{eq:Z_lebesgue_new}
\end{align}
Nested sampling proceeds by constructing likelihood contours and assigning prior volumes to them.
The computational difficulty is not the one-dimensional quadrature once these volumes are known, but the stochastic inference of the volume sequence.

The present paper reframes nested sampling so that the volume sequence is treated as a random variable with a race-induced prior that can be conditioned on auxiliary Monte Carlo information.
The auxiliary information comes from \textit{phantom samples}: intermediate states produced by a Markov-chain-type constrained-prior sampler while generating a \textit{classic sample}.
The method therefore does not apply to rejection-sampling constrained-prior approaches, whose rejected proposals are not draws from the constrained prior.
Phantom samples are correlated and thus cannot be used in the classic shrinkage formalism.
The reason can be understood in two ways.
Under the race view, it is because the races are no longer independent, so winner rate is no longer the sum of all rates.
Under the order-statistic view, it is because the correlation gives non-uniform probability of some permutations of orderings of likelihoods.
In this paper, we will adopt the race view, which is described in Section~\ref{sec:races}.

Despite their inadequacy for classic shrinkage, phantom samples are still drawn from the same target distribution, and thus can be used for unbiased Monte Carlo expectation estimators.
We will use them to form unbiased estimates of constrained-prior probabilities,
\begin{align}
    \Prob_{x\sim p(\cdot\mid L>a)}(L(x)>b)=\frac{X_b}{X_a},\label{eq:intro_phantom_identity}
\end{align}
which are precisely the shrinkage ratios supplied by the classic order statistics in NS.
Correlated phantom states can be valid Monte Carlo observations of an expectation while being invalid as race participants.
To our knowledge, phantom states have not previously been used in this direct way to improve the evidence calculation itself.

The contributions are as follows.
\begin{itemize}
    \item First, we formulate NS shrinkage as exponential races of lineages, introducing the \textit{race tree}, which leads to a simplified dynamic nested sampling formulation.
    \item Second, we present a Bayesian formalism for NS which properly treats plateaus.
    \item Third, building on this Bayesian presentation, we introduce phantom-conditioned NS.
    \item Finally, we present JAXNS v3, an open-source implementation of phantom-conditioned NS, validate it, and show its limitations.
\end{itemize}

The paper is structured in the following order.
Section~\ref{sec:strict_ns} introduces the race tree formulation.
Section~\ref{sec:bayesian_shrinkage} introduces a Bayesian model of which handles plateaus rigorously.
Section~\ref{sec:phantoms} uses the Bayesian model to develop phantom-conditioned NS.
Section~\ref{sec:v3} describes the JAXNS v3 implementation choices.
Section~\ref{sec:experiments} presents paired evidence-accuracy benchmarks of classic and phantom-conditioned shrinkage.

\section{Classic NS as an exponential race}
\label{sec:strict_ns}\label{sec:races}

For a threshold $\lambda$, define the strict constrained set and its prior volume by,
\begin{align}
    \mathcal{A}_\lambda &\triangleq \{x\in\mathcal{A}: L(x)>\lambda\},\\
    X_\lambda &\triangleq p(\mathcal{A}_\lambda).
    \label{eq:X_lambda}
\end{align}
When $X_\lambda>0$, define the strict constrained prior by,
\begin{align}
    \pi_\lambda(\dd x)
    \triangleq
    \frac{\ind\{L(x)>\lambda\}}{X_\lambda}\,p(\dd x).
    \label{eq:strict_constrained_prior}
\end{align}
We shall call a \textit{sample} from a constrained prior a pair $(x_i \sim \pi_\lambda, L_i\triangleq L(x_i))$, with \textit{sample index} $i$ the label of the sample.
% Note to self for future exploration:
% All shrinkage quantities in this paper are strict.
% Inclusive constraints may be useful as a constrained-sampling device for exploring a plateau, but they do not define a second shrinkage law in the evidence calculation.

Let $\lambda_0$ be a \textit{sentinel} likelihood associated with $X_{\lambda_0}=1$.
By convention, $\lambda_0$ is often set to $0$.
For increasing distinct likelihood levels,
\begin{align}
    \lambda_0 < \lambda_1 < \cdots < \lambda_G,\label{eq:strict_likelihooods}
\end{align}
write $X_g\triangleq X_{\lambda_g}$.
Nested sampling quadrature estimates the evidence as,
\begin{align}
    Z \approx \sum_{g=1}^{G} \lambda_g\,(X_{g-1}-X_g).
    \label{eq:block_quadrature}
\end{align}
The likelihood levels in Eq.~\ref{eq:strict_likelihooods} are an arbitrary \textit{schedule} of values, and the goal of NS is to assign each level $\lambda_g$ a random variable $X_g$ with known distribution.
We call the subscript $g$ the \textit{block index} associated with the level $\lambda_g$.
In addition, an auxiliary output of NS is a set of weighted samples drawn from the posterior.
In the absence of likelihood plateaus, we can uniquely associate each block index with a sample index, $g(i)$, and the unnormalised posterior sample weights are,
\begin{align}
    w_{g(i)} \triangleq \lambda_{g(i)}\, (X_{{g(i)}-1}-X_{g(i)}).\label{eq:post_weight_cont}
\end{align}
The posterior weights for samples that fall in likelihood plateaus will be shown later.

Let us introduce the concept of a race tree.
Let there be a set of $N$ samples, as well as a sentential sample with sample index $0$ given likelihood value $L_0 = \lambda_0$.
These samples form a race tree if each sample (excluding the sentinel sample) is an independent draw from some constrained prior whose contour is also contained in the set of samples,
\begin{align}
    \forall i \in& (1 \ldots N)\\
    \exists\, p(i) \in& (0 \ldots N) \label{eq:parent_def}\\
    x_i \sim& \,\pi_{\lambda_{p(i)}}.
\end{align}
In this case we say that $p(i)$ is the parent of $i$, denoting this as $p(i) \to i$, and that these samples form a race tree.
A \textit{lineage} is a maximal chain of parent links.
We shall refer to these samples as \textit{classic samples}, to differentiate them from phantom samples later.
The samples are non-decreasing,
\begin{align}
     \lambda_0 = L_0<L_1 \le \ldots \le L_N,
\end{align}
and we define a \textit{block} for a level $\lambda_g$ as the set of all sample indices with likelihood equal to $\lambda_g$,
\begin{align}
    \mathcal{B}_g &\triangleq \{i:L_i=\lambda_g\},
    \label{eq:block_def}\\
    m_g &\triangleq |\mathcal{B}_g|.
\end{align}
A block with $m_g > 1$ represents a plateau, where the ordering of samples in the block is not unique.

The out-degree of sample $i$ is,
\begin{align}
    d_i \triangleq \#\{j:p(j)=i\}.
\end{align}
The sentinel must have out-degree $d_0 > 0$ in order for the race tree to not be empty.

The crucial observation in NS is that a single sample generated from a parent $p(i) \to i$ shrinks the prior volume uniformly,
\begin{align}
    r_{p(i) \to i} \triangleq \frac{X_{\lambda_{i}}}{X_{\lambda_{p(i)}}} \sim&\, \mathcal{U}[0,1]\\
    \Longleftrightarrow\quad -\log r_{p(i) \to i} \sim&\, \mathrm{Exp}(1).
\end{align}
Define \textit{race time} $s \triangleq -\log X$ and \textit{race rank} $r\triangleq \exp -s$, and interpret the race tree as a set of lineages equipped with unit-rate clocks incrementally racing from one sample to the next sample.
For any sample-to-sample race $i-1 \to i$ there are $K_i$ active racing lineages.
We have that the race time for $i-1 \to i$ is,
\begin{align}
    s_{i-1 \to i} \sim& \mathrm{Exp}(K_i)\label{eq:race_increment}\\
    r_{i-1 \to i} =& e^{-s_{i-1 \to i}}.\label{eq:beta_shrinkage}
\end{align}
As is well-known from the order-statistic treatment of NS, $r_{i-1 \to i} \sim \mathrm{Beta}(K_i, 1)$.
Starting with $K_1=d_0$, after the race $i-1\to i$ is done the active lineage count updates as,
\begin{align}
    K_{i+1} = K_i - 1 + d_{i}.\label{eq:K_recursion}
\end{align}

An interesting aspect of this formulation is that there is no notion of live points, nor of a specific order that samples must be generated, beyond the race tree structure requirements.
A new child can be generated from any existing strict contour, and the only thing that must be tracked is the out-degree of its parent.
This opens up the potential for parallel processing novelties.

\begin{proposition}[Race-tree shrinkage]\label{prop:race_tree}
Suppose that there are $K_i$ active lineages participating in race $i-1 \to i$, each represented by a unit-rate clock.
Then $s_{i-1\to i}\sim\Exp(K_i)$, and afterwards the rate update is updated by Eq.~\ref{eq:K_recursion}.
\end{proposition}

\begin{proof}
The minimum of independent exponential clocks is exponential with rate equal to the sum of the rates.
Once the arriving sample $i$ is observed, its lineage segment no longer contributes to the next segment of the race and its $d_i$ child lineage segments begin.
This changes the active lineage count by $-1+d_i$, giving Eq.~\ref{eq:K_recursion}.
\end{proof}

In classic races there are no ties, however in the presence of a plateau, $m_g>1$, we do have ties.
The correct way to handle this is to introduce a nuisance winner ordering parameter that decides the order of winners, and that the race is not over until all winners have been accounted for.
While the race is proceeding the race clocks are still ticking, hence each sample $k$ of block $\mathcal{B}_g$ gets a different race rank, $r^{(k)}_{g-1 \to g}$ falling within the atom $L = L_g$.
It can easily be shown that the final sample in the plateau has known rank,
\begin{align}
    r^{(m_g)}_{g-1 \to g} \sim\, \mathrm{Beta}(K_g - m_g + 1, m_g).\label{eq:end_block_dist}
\end{align}
However, since all plateau samples are strictly inside the atom, none of them correspond exactly to the end of the block, i.e. we have that $r^{(m_g)}_{g-1 \to g} \neq r_{g-1 \to g}$ almost surely.
This poses a problem because in order for shrinkage to proceed we need to know the race rank at the exact end of a block, not merely the race rank of the last sample in the block.
This is depicted in Figure~\ref{fig:censor}.

It will be helpful to define the following quantities,
\begin{align}
    p_{<g} \triangleq&\, \Prob(L_g > L > L_{g-1})\label{eq:open_category_defs_race}\\
    p_{=g} \triangleq&\, \Prob(L = L_g)\label{eq:atom}\\
    p_{>g} \triangleq&\, \Prob(L > L_g) \label{eq:strict_endpoint_pg} \\
    =&\, \frac{X_g}{X_{g-1}}\\
    =&\,r_{g-1 \to g},
\end{align}
where Eq.~\ref{eq:open_category_defs_race} defines the race interval, Eq.~\ref{eq:atom} defines a possible atom, and Eq.~\ref{eq:strict_endpoint_pg} is the end point of the race.
As shown in Figure~\ref{fig:censor} the presence of an atom censors us from learning the end point $p_{>g}$.
Thus, when there is a plateau we know that all the plateau samples \textit{finish} the race, just not \textit{when},
\begin{align}
    p_{>g} \le r^{(m_g)}_{g-1 \to g} \le p_{>g} + p_{=g}.
\end{align}
We call this loss of knowledge \textit{plateau censorship}.

\begin{figure*}[ht]
    \centering
    \IfFileExists{images/censor.drawio.pdf}{%
        \includegraphics[width=0.9\linewidth]{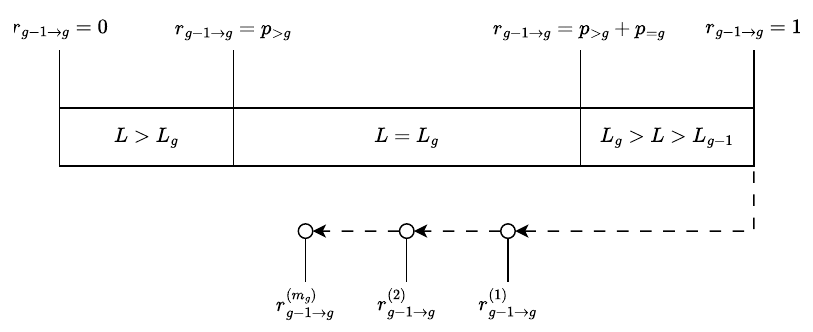}%
    }{%
        \placeholder{Plateau-censorship diagram: the source asset
        \texttt{images/censor.drawio.pdf} is not present in this checkout.}%
    }
    \caption{The presence of an atom $L=L_g$ causes us to to lose exact knowledge of $p_{>g}$, since $r^{(m_g)}_{g-1 \to g} \in [p_{>g}, p_{>g} + p_{=g}]$. The atom can be seen as an interval in race rank just before finishing the race where samples loose identity. None of the samples within the atom correspond exactly with the end of the block, whereas when $m_g=1$ it is exactly equivalent to the rank of the end of the block.}
    \label{fig:censor}
\end{figure*}

\section{Bayesian shrinkage model}
\label{sec:bayesian_shrinkage}

We propose a Bayesian model that will simultaneously resolve the problem of plateau censorship, and allow us to condition shrinkage on phantom samples.
Define the block probability vector,
\begin{align}
    p_g = (p_{>g}, p_{=g}, g_{<g}),
\end{align}
where clearly $|p_g| = 1$.
Consider a block with $K_g$ incoming active lineages and $m_g>1$ block samples.
Then, these define three observation categories,
\begin{itemize}
    \item $K_g - m_g$ samples lie past the atom,
    \item $m_g$ samples fall in the atom,
    \item $0$  samples fall in the race interval.
\end{itemize}
These are independent categorical observations, therefore the likelihood of $p_g$ is multinomial,
\begin{align}
    \mathcal{L}(p_g) \propto p_{>g}^{K_g - m_g} p_{=g}^{m_g}.
\end{align}
We set a Dirichlet prior,
\begin{align}
    p(p_g) = \mathrm{Dirichlet}(1, \epsilon_g, 1 - \epsilon_g),
\end{align}
which can be done be identifying the marginal of the Dirichlet, $\mathrm{Beta}(a_i, \sum_{j \neq i} a_j)$, with uniform shrinkage prior, $\mathcal{U}[0,1] \equiv \mathrm{Beta}(1, 1)$.
The variable $\epsilon_g$ is a free parameter and should represent our belief of $\E[p_{=g}/(1 - p_{>g})]$.
A fairly neutral choice would be $\epsilon_g=1/2$, given that $m_g > 1$.
Using conjugacy of the Dirichlet and multinomial distributions the posterior becomes exactly,
\begin{align}
    &p_g \mid K_g, m_g \sim\, \mathrm{Dirichlet}(a_{>g}, a_{=g}, a_{<g}),\label{eq:post_pg}\\
    &a_{>g}= K_g - m_g + 1,\\
    &a_{=g}= m_g + \epsilon_g,\\
    &a_{<g}= 1 - \epsilon_g
\end{align}
therefore the posterior marginals for $p_{>g}$ and $p_{=g}$ are,
\begin{align}
    p_{>g}  \mid K_g, m_g \sim&\, \mathrm{Beta}(K_g-m_g+1,m_g+1)\label{eq:posterior_pg}\\
    p_{=g}  \mid K_g, m_g \sim&\, \mathrm{Beta}(m_g+\epsilon_g,K_g-m_g+2-\epsilon_g)\label{eq:posterior_pe}
\end{align}
Comparing Eq.~\ref{eq:posterior_pg} to Eq.~\ref{eq:end_block_dist}, we see the difference, which proves censorship in a rigorous way.

Evidence quadrature, Eq.~\ref{eq:block_quadrature}, is be done using,
\begin{align}
    X_g = X_{g-1} p_{>g},
\end{align}
which can be used for form evidence expectations, or to produce samples of the evidence.
Now, coming back to assigning posterior mass to plateau samples, this this is done by equally assigning atom mass to each sample $k\in\mathcal{B}_g$,
\begin{align}
    w^{(k)}_{g} = \frac{\lambda_g}{m_g} X_{g-1} p_{=g},
\end{align}
whereas the posterior mass assigned to a non-plateau sample, mirroring Eq.~\ref{eq:post_weight_cont}, is,
\begin{align}
    w_g = \lambda_g X_{g-1}(1 - p_{>g}).
\end{align}
In general, the posterior sample weights do not sum to the evidence quadrature when plateaus are present because $p_{=g}\neq 1-p_{>g}$.
The difference is the open race-interval mass $p_{<g}$.

When $m_g=1$, clearly Eq.~\ref{eq:posterior_pg} does not match Eq.~\ref{eq:end_block_dist}.
This is reconciled by observing that a singleton block provides no evidence for an atom, so we must choose a two-class rather than a three-class prior.
Removing the equality category recovers the exact classic shrinkage,
\begin{align}
    p_{>g} \mid K_g \sim\, \mathrm{Beta}(K_g,1),
\end{align}
when $m_g=1$.

\section{Phantom-conditioned shrinkage}
\label{sec:phantoms}

We are now equipped to use information from phantom samples to improve shrinkage.
Let $\mathcal{P}_i$ be phantom samples associated with sample $i$, which we'll call a \textit{cluster} of phantom samples.
They are correlated draws from $\pi_{\lambda_{p(i)}}$ produced in the process of realising sample $i$.
Typically, this is because a Markov-chain approach is used to produce iid samples from $\pi_{\lambda_{p(i)}}$.
Since phantom samples are not independent they cannot participate in the race tree, however we are perfectly well able to use them to construct consistent unbiased MC estimates of expectations under $\pi_{\lambda_{p(i)}}$.
We thus seek a procedure for sampling from $p_g \mid K_g, m_g, \bigcup_i \mathcal{P}_i$.
As above, we first work under the assumption that a plateau is present, so $m_g>1$.
We then derive the $m_g=1$ case in the same manner, using the correct limiting two-class prior when there is no evidence for an atom.

We assume that phantom samples \textit{between} clusters are approximately independent, while allowing arbitrary correlation within each cluster. Drawing initial chain points from eligible stationary seeds does not by itself establish this independence: shared ancestry and finite constrained-chain exploration can induce dependencies. The resulting uncertainty therefore requires empirical calibration, which we assess in Section~\ref{sec:experiments}.
In fact, it's possible to redefine phantom clusters as any disjoint decomposition of phantom samples where each cluster is drawn stationarily from within the same parent, and samples are independent between clusters.
As will be shown below, it is best to make clusters as small as possible such that this is true.

Lets define, per cluster $c$, block interval counts,
\begin{align}
    A_{cg} =&\, \sum_{(x, L) \in \mathcal{P}_c} \mathbf{1}\{L(x) > L_{g-1}\} \mathbf{1}\{L_{c} \le L_{g-1}\}\\
    B_{cg} =&\, \sum_{(x, L) \in \mathcal{P}_c} \mathbf{1}\{L(x) > L_{g}\} \mathbf{1}\{L_{c} \le L_{g-1}\}\\
    E_{cg} =&\, \sum_{(x, L) \in \mathcal{P}_c} \mathbf{1}\{L(x) = L_{g}\} \mathbf{1}\{L_{c} \le L_{g-1}\}.
\end{align}
For each cluster, these are counts of samples falling in each interval in Figure~\ref{fig:censor}.
Then an MC estimate of $p_g$ is,
\begin{align}
    \hat{p}_{g} =& \frac{\sum_c \left( B_{cg}, E_{cg}, A_{cg} - B_{cg} - E_{cg}\right)^\top}{\sum_c A_{cg}}.
\end{align}
Now, if the phantom samples were independent then the observable $\sum_c A_{cg} \hat{p}_g$ would satisfy the multinomial likelihood.
In this case, the exact same machinery as in Section~\ref{sec:bayesian_shrinkage} could be used, and phantom-conditioned posterior shrinkage would be a Dirichlet distribution.
However, as they are not independent the posterior has no closed form.

One idea is to use a multinomial likelihood with an effective count calibrated for within-cluster correlation.
However, a scalar effective count wrongly accounts for correlations in joint shrinkage and is sensitive to how it is calibrated.
This approach was initially attempted and rejected for these reasons.

We instead propose an MC-based approach that asymptotes to the Dirichlet posterior in the limit of independent phantom samples.
To do this we will use the gamma construction of Dirichlet random variables to rewrite the race-induced posterior, Eq.~\ref{eq:post_pg},
\begin{align}
    M_{>g} \sim&\, \mathrm{Gamma}(a_{>g}, 1)\\
    M_{=g} \sim&\, \mathrm{Gamma}(a_{=g}, 1)\\
    M_{<g} \sim&\, \mathrm{Gamma}(a_{<g}, 1)\\
    \frac{(M_{>g}, M_{=g}, M_{<g})^\top}{M_{>g} + M_{=g} + M_{<g}} \sim&\, \mathrm{Dirichlet}(a_{>g}, a_{=g}, a_{<g}).\label{eq:gamma_dir}
\end{align}
Now consider how we would condition on iid phantom samples.
In this case, we are free to decompose each phantom cluster into singletons because doing so preserves inter-cluster independence.
Then, we would have that $A_{cg}, B_{cg}, E_{cg} \in \{0, 1\}$, that is each cluster would contribute at most one count observation to each block.
Define a cluster weight $v_c$ and then define,
\begin{align}
    M'_{>g} =& M_{>g} + \sum_c v_c B_{cg}\\
    M'_{=g} =& M_{=g} + \sum_c v_c E_{cg}\\
    M'_{<g} =& M_{<g} + \sum_c v_c (A_{cg} - B_{cg} - E_{cg})\\
    p_g =& \frac{(M'_{>g}, M'_{=g}, M'_{<g})^\top}{M'_{>g} + M'_{=g} + M'_{<g}}.
\end{align}
Using additivity of identical-scale gamma random variables, we can then identify the natural cluster weight as,
\begin{align}
    v_c \sim\, \mathrm{Gamma}(1, 1).
\end{align}
With this choice of cluster weight, singleton clusters exactly produce the Dirichlet posterior in Eq.~\ref{eq:post_pg}.

Now, consider what happens when cluster sizes are not singleton, that is we have a non-trivial irreducible correlation structure.
In this case, the gamma's do not add as the scales are not identical.
Under multiplication gamma variance scales quadratically, compared to linearly under addition, whereas the mean scales linearly under both multiplication and addition.
Thus, the net effect of correlated clusters is to inflate variance, while preserving the mean, which is precisely what we expect correlation to do.
Hence, if a single cluster can be split into two disjoint clusters where independence between the two clusters is approximately true, then it is favourable from the perspective of information utilisation to do this.

An important well-known consideration for conditioning on MC counts is that there should be enough independent counts to properly capture the variance of outcomes.
This why using sparsely populated histograms to draw conclusions is a bad idea.
We impose a minimum participating cluster count criterion in order for a block to condition on phantom count information.
Specifically, we use a Kish-like quantity to measure effective participating cluster count,
\begin{align}
    C^{\rm Kish}_g = \frac{\left(\sum_c A_{cg}\right)^2}{\sum_c A_{cg}^2},
    \label{eq:kish_cluster_count}
\end{align}
and then define a minimum threshold for condition, say $C_{\rm min} = 20$, and hence,
\begin{align}
    M'_{>g} =& M_{>g} + \mathbf{1}\{C^{\rm Kish}_g \ge C_{\rm min}\} \sum_c v_c B_{cg}\\
    M'_{=g} =& M_{=g} + \mathbf{1}\{C^{\rm Kish}_g \ge C_{\rm min}\} \sum_c v_c E_{cg}\\
    M'_{<g} =& M_{<g} + \mathbf{1}\{C^{\rm Kish}_g \ge C_{\rm min}\} \sum_c v_c (A_{cg} - B_{cg} - E_{cg})\\
    p_g =& \frac{(M'_{>g}, M'_{=g}, M'_{<g})^\top}{M'_{>g} + M'_{=g} + M'_{<g}}.
\end{align}
This provides a consistent generative shrinkage model that properly takes into account the race tree and phantom samples.
The shrinkage at block $g$ is the the component $p_{>g}$, while the plateau atom size is $p_{=g}$.
It is simple to calculate as the per-cluster counts are computed once, and then each draw from the joint posterior is just a set of draws from gamma distributions followed by reduction.

We now take the prior in the limit of no atom-mass evidence, which is appropriate when $m_g=1$ and there is no other evidence for an atom.
The equality category is then structurally absent,
\begin{align}
    M_{=g} = M'_{=g} = p_{=g} = 0,
\end{align}
and the race remains two-class.
The phantom-conditioned construction then becomes,
\begin{align}
    M_{>g} \sim&\, \mathrm{Gamma}(K_g, 1),\\
    M_{<g} \sim&\, \mathrm{Gamma}(1, 1),\\
    M'_{>g} =&\, M_{>g} + \mathbf{1}\{C^{\rm Kish}_g \ge C_{\rm min}\} \sum_c v_c B_{cg},\\
    M'_{<g} =&\, M_{<g} + \mathbf{1}\{C^{\rm Kish}_g \ge C_{\rm min}\} \sum_c v_c (A_{cg} - B_{cg}),\\
    (p_{>g},p_{<g})^\top =&\, \frac{(M'_{>g}, M'_{<g})^\top}{M'_{>g} + M'_{<g}}.
\end{align}
As a sanity check, we verify that in the limit of iid phantom samples this construction recovers classic shrinkage,
\begin{align}
    p_{>g} \mid \bigcup_c \mathcal{P}_c
    \sim \mathrm{Beta}\left(
        K_g + \sum_c B_{cg},
        1 + \sum_c (A_{cg} - B_{cg})
    \right).
\end{align}
Thus, with no participating phantom clusters or plateaus, it recovers exactly $p_{>g}\sim\mathrm{Beta}(K_g,1)$.

The above formulation assumes that the phantom samples in $\mathcal{P}_i$ are draws from $\pi_{\lambda_{p(i)}}$.
Thus, when they are produced by a Markov chain, the chain must be stationary enough that the retained phantom states have the correct marginal expectation under the parent contour.
The formulation cannot correct for non-stationary phantom samples.
We address this in Section~\ref{sec:v3}.

\section{JAXNS v3 implementation}\label{sec:v3}

The race tree permits new children to be generated from any existing strict contour. We exploit this by representing allocation work as logical \textit{lineage threads}, which is an idea initially proposed in \citep{2019S&C....29..891H}.
A thread is a single lineage from an initial parent, where the out-degree of each sample is one.
An edge $a \to b$ adds one active lineage to every existing block $g$ whose contour it crosses, \begin{align}
L_{a}<L_g\le L_b. \label{eq:edge_lineage_coverage}
\end{align}
Consequently, a completed thread beginning immediately before block $a$ and terminating once it reaches block $b$ adds one active lineage to every block $g\in\{a,\ldots,b\}$.
We'll denote such a thread by $T(a,b)$.
The final child may overshoot $L_b$, in which case the thread also provides additional lineage allocation beyond its requested terminal block.
A thread is logical in nature, and thus its successive edges may be generated at different scheduling times.
In particular there is no need for a scheduling barrier.
JAXNS uses this property for distributed processing by treating cluster workers as a pool to which sampling work is outsourced while scheduling and thread evolution overlap.

Algorithm~\ref{algo:jaxns} shows the core algorithm of JAXNS.
There is an outer loop which, until a goal condition is met, iteratively sets a lineage allocation target and calls an inner loop that schedules thread sampling according to that target.
The inner loop decomposes the target into threads and advances them until a depth condition is met.

The depth conditions available in JAXNS are,
\begin{align}
    \text{small remaining evidence}\quad
    \frac{L_G X_G}{\sum_g^G L_g(X_{g-1} - X_g) + L_G X_G}
    &< \tau_Z,
    \label{eq:small_remaining_evidence}
    \\
    \text{small remaining posterior mass}\quad
    \frac{L_gX_g}{\max_{g'\le g}(L_{g'}X_{g'})}
    &< \tau_{\rm post}.
\end{align}
The goal condition is user specified, and can be any condition computable from the samples, for example when a target effective sample size or evidence uncertainty is met.

\subsection{Allocation targets}
Given a lineage count $K_g$ for each block and an allocation target $D_g^k$ at allocation iteration $k$, we define the allocation gap,
\begin{align}
    G_g^k = (D_g^k - K_g)_+,
    \label{eq:allocation_gap}
\end{align}
which is equivalent to the number of lineage threads that must contain block $g$.
A schedule of sampling is done where $G_g^k$ is decomposed into a set of threads.
This decomposition can be done in any manner.
To use likelihood evaluations efficiently, JAXNS greedily decomposes $G_g^k$ into maximal threads.
This decomposition uses the fewest threads needed to fill the allocation gap.

JAXNS provides uniform, evidence-improving, and posterior-improving allocation targets.
At each goal iteration, uniform allocation assigns the same target lineage count to every contour, recovering the allocation of static nested sampling with a fixed number of live points.
The uniform allocation target $D_g^k=d_0+\Delta K k$ leads to the allocation gap
\begin{align}
    G_g^k &= (d_0+\Delta K k-K_g)_+,
    \label{eq:uniform_allocation_gap}
\end{align}
whereas the evidence- and posterior-improving allocation targets are based on unit-peak utilities derived in Appendix~\ref{sec:allocation_targets}.
Writing $\bar U_g$ for the relevant utility, they give the allocation gap
\begin{align}
    G_g^k &= \left\lceil\Delta K\bar U_g\right\rceil.
    \label{eq:utility_allocation_gap}
\end{align}
Here $\Delta K>0$ controls the amount of work introduced by one goal iteration.

\subsection{Thread scheduling}

At the start of a scheduling round, JAXNS computes the allocation gap on the existing likelihood blocks and decomposes it into maximal threads.
The requested thread intervals and allocation target remain fixed while the thread heads advance.
New threads and continuations are selected in order of increasing effective parent contour, with new threads taking precedence at equal contours.
Each thread finishes when its child reaches or exceeds its terminal contour.
The population of available seeds can grow during the round without changing the requested thread intervals or allocation target.
When all threads in the round have finished, the likelihood blocks and expected volume path are updated, and the depth condition is evaluated.
If the depth condition is not met, JAXNS recalculates the allocation gap on the updated likelihood blocks using the same target.
It then generates further samples wherever additional lineages are needed, without evaluating the goal condition.

\begin{algorithm}
\caption{The JAXNS core algorithm with dynamic lineage allocation, outer goal terminating loop, and inner depth termination loop.}
\begin{algorithmic}[1]
\Require root lineage count $d_0$, sentinel likelihood $\lambda_0$
% \ENSURE $\{(x_i, L_i, \mathcal{P}_i, d_i): i=1 \ldots N\}$

\State Sample $(x_i, L_i, \mathcal{P}_i) \sim \pi_{\lambda_0}$ for $i\in\{1 \ldots d_0\}$
\State Set $N=d_0$
\State Set allocation iteration $k=0$

\Repeat
    \Repeat
        \State Compute $G_g^k$ using Eq.~\ref{eq:uniform_allocation_gap} or Eq.~\ref{eq:utility_allocation_gap}.
        \State Decompose $G_g^k$ into a set of maximal threads, $\mathcal{T}$.
        \While{$\mathcal{T} \neq \varnothing$}
            \State Select a thread, $T \in \mathcal{T}$.
            \State Set $h \to \mathrm{head}(T)$.
            \State Set $i \to p(h)$.
            \State Set $N \to N + 1$
            \State Set $j \to N$
            \State Choose a stationary seed $x^0_j$ for $\pi_{L_i}$ from the available classic samples.
            \State Sample $(x_j, L_j, \mathcal{P}_j) \sim\, \pi_{L_i}$ from $x^0_j$.
            \State Set $d_i \to d_i + 1$
            \State Set $\mathrm{head}(T) \to j$.
            \If{$L_j > L_{\sup(T)}$}
            \State Remove $T$ from $\mathcal{T}$.
            \EndIf
        \EndWhile
    \Until{$\mathbf{depth}$}
    \State Set $k\to k+1$.
\Until{$\mathbf{goal}$}

\State \Return $\{(x_i, L_i, \mathcal{P}_i, d_i): i\in\{1 \ldots N\}\}$

\medskip
\Statex \textbf{Note 1:} If no stationary seed exists there for a requested contour, the thread is reparented to the closest shallower contour with a stationary seed, or to the sentinel if necessary.
\end{algorithmic}
\label{algo:jaxns}
\end{algorithm}

\subsection{Constrained-prior sampling}

\begin{algorithm}
\caption{The JAXNS constrained sampler using one-dimensional slice sampling with perfect bracketing and uniform trajectory sampling.}
\begin{algorithmic}[1]
\Require direction kernel $\mathcal{D}$, contour $\lambda$, seed point $x^0$, steps per acceptance $T$.
\State Set $i=0$.
\While{$i < T$}
    \State Choose direction, $\hat{n}\sim\,\mathcal{D}$.
    \State Construct a trajectory $I$ that passes through $x^i$ along $\hat{n}$, with end points outside the slice and independent from $x^i$.
    \State Uniformly sample $x^{i+1}\sim\,\pi_\lambda$ along $I$.
    \State Set $i \to i + 1$.
    \State Set $\mathbf{done}=L(x^i) > \lambda$.
\EndWhile
\State Set phantom cluster $\mathcal{P} = \{(x^k, L(x^k)) : k \in \{1 \ldots T-1\}\}$.
\State \Return $(x^T, L(x^T), \mathcal{P})$

\medskip
\end{algorithmic}
\label{algo:constrained_sampler}
\end{algorithm}

Nested sampling requires each classic child to be approximately sampled from the strict constrained prior of its parent.
The JAXNS constrained sampler uses one-dimensional slice sampling \citep{2000physics...9028N} from a seed point $x^0$, as shown in Algorithm~\ref{algo:constrained_sampler}.
The seed $x^0$ must itself be distributed according to the parent's strict constrained prior, so that the chain starts in stationarity.
Each Markov transition explores the constrained prior, with a total of $T$ steps per chain.
Because the seed is stationary and each transition preserves the constrained prior, no burn-in is required.
All generated chain states except the final classic sample form the phantom cluster.
A single step involves choosing a direction from a symmetric distribution $\hat{n} \sim\, \mathcal{D}$ independent of the current chain.
That is, all directions in the chain should come from the same distribution.
Given a direction $\hat{n}$ and chain point $x^i$ a trajectory that brackets $x^i$ along the direction is chosen.
The trajectory must have end points that land outside the contour, and the end points should be independent from $x^i$.
Then, constrained sampling is performed along this trajectory, which can be done in any manner.
A common approach is to randomly pick a point along the trajectory and greedily shrink the trajectory if the point falls outside the slice.

The default direction proposal distribution $\mathcal{D}$ is an isotropic Gaussian in homogeneous prior space, which is the unit hypercube in JAXNS.
This is the direction law used in all experiments in Section~\ref{sec:experiments}.
JAXNS also provides Gaussian-mixture-model (GMM) directions.
A Gaussian-mixture likelihood surrogate is fitted in homogeneous prior space using all classic samples collected so far and their stored likelihoods.
At direction-sampling time, an ellipsoid is selected in proportion to its fitted volume above the parent contour $\lambda$, and its covariance is then used to sample a slice direction.
This choice makes the direction distribution independent of the current chain point while maximising the probability of selecting a direction relevant to the constrained region.
Isotropic search directions are randomly injected to prevent collapse of the direction distribution.

JAXNS uses perfect bracketing along a straight line to construct the bracketing trajectory $I$ in Algorithm~\ref{algo:constrained_sampler}.
The endpoints are the intersections of the line with the bounds of the unit hypercube representing the homogeneous prior measure.
This is the maximal interval possible and can be determined without any likelihood evaluations.
Constrained sampling along this trajectory proceeds with uniform sampling and greedy trajectory shrinkage which exponentially contracts towards the current chain point, where the slice is necessarily satisfied.
This approach trades off the zero-cost trajectory determination with extra likelihood evaluations to uniformly sample the trajectory.
Mode evaporation remains possible with this approach.
The probability of tunnelling between modes depends on the geometry of the problem.
Mode-mass recovery and comparisons with other samplers are assessed in Section~\ref{sec:mode_death_comparison}.

\subsection{Seed selection}

The seed point must be chosen to maintain stationarity of the chain without requiring burn-in.
Specifically, if sampling from parent $j$ with likelihood $L_j$, i.e. sampling from $\pi_{L_j}$, then we should choose a seed as some sample $i_0$ already satisfying stationarity.
That is, choose $i_0$ such that $L_{p(i_0)} \le L_j < L_{i_0}$.

JAXNS selects constrained-chain seeds exclusively from the eligible classic samples available at the time of selection, and choices at different contours may reuse the same sample.
Phantom samples are not included in the eligible seed population because their reuse may increase inter-cluster correlation, as discussed in Section~\ref{sec:phantom_seed_dependence}.

Efficient seed selection must make newly generated classic children available without reconsidering the entire growing sample collection after every group of chains.
At the start of a scheduling round, every classic sample collected so far is available, subject to the eligibility condition above.
Between full updates of this population, JAXNS makes a bounded random subset of newly generated classic children available, with membership chosen independently of their likelihoods.
This allows recent children to seed later chains while limiting the work needed to incorporate new seeds.
All accumulated classic children are made available together as the collection grows and before each new scheduling round.

A selected parent may have no available stationary seed in the seed population.
In this case, the thread is reparented to the nearest shallower contour with eligible seeds, falling back to the sentinel node if necessary, and an eligible seed is selected at random.

\subsection{Massively scalable distributed processing}

JAXNS is designed for heterogeneous distributed likelihood evaluation scaling from a laptop to clusters with thousands of nodes.
A central process owns the scientific state and thread schedule, while workers perform likelihood evaluations and constrained-prior sampling.
The only network topology requirement is that workers can initiate connections to the coordinator's exposed port, so connection authorisation is needed only in that direction.
Workers receive serialised likelihood models, arguments, parameters, and sampler configuration.
Communication uses ZMQ and serialisation uses Python's \texttt{pickle} functionality combined with JAX's Pytree representation.
Model arguments may consist of arbitrarily nested trees of literals and arrays.

The thread schedule exposes the heads of incomplete logical threads as independent sampling work from potentially different parent contours.
The coordinator assigns these heads across all available worker execution lanes and replenishes capacity as completed work is incorporated, provided that enough incomplete threads are available.
This allows constrained-prior sampling calls to overlap without requiring every thread to reach the same contour or terminate at a scheduling barrier.
The race-tree formulation makes this separation of logical work from physical workers straightforward because a completed child only increments the out-degree of its known parent and advances its logical thread.

The worker pool is elastic throughout a run.
Workers may register after sampling has begun, including after a node has been restarted, and immediately contribute their available capacity.
Heartbeats identify workers that have become unavailable and fence their expired leases.
Their unfinished tasks are returned to the queue, while the corresponding node supervisor is asked to restart the failed worker.
If the entire pool becomes unavailable, the scientific process retains its state and waits for a worker to recover or join rather than terminating the run.
Consequently, workers may join or leave at any time without making cluster membership part of the scientific algorithm.
Checkpointing and resumability also allow users to inspect intermediate results, change goal targets, and steer the sampling process interactively.

\section{Results}\label{sec:experiments}

We now provide a set of experiments aimed at validating phantom conditioning, as well as determining its limitations.
In addition, we characterise the evidence-improving and posterior-improving allocation schemes.
We provide numbered claims for each finding.
We introduce three problem models with reference evidence estimates that are used to substantiate each claim.
For each substantiation we run the model with $30$ different random seeds.

\subsection{Known-evidence problems}

The three baseline problems have dimension $D=10$ and vary contour curvature and multimodality while retaining deterministic reference evidences.
Below $\phi_D(x\mid\mu,\Sigma)$ denotes a $D$-variate normal density.
Figure~\ref{fig:evidence_problems} shows their two-dimensional analogues.

\paragraph{Correlated Gaussian (G10).}
This problem places a shifted, strongly correlated Gaussian likelihood in the tail of an isotropic Gaussian prior,
\begin{align}
    L_{\rm G10}(x)=\phi_{10}(x\mid\mu_{\rm G},\Sigma_\ell),\qquad x\sim\mathcal{N}(0,I).
    \label{eq:g10_likelihood}
\end{align}
Here $\mu_{\rm G}=2e_1-3e_2=(2,-3,0,\ldots,0)$, $(\Sigma_\ell)_{ii}=1$, and $(\Sigma_\ell)_{ij}=0.99$ for $i\ne j$.
Gaussian conjugacy gives
\begin{align}
    \log Z_{\rm G10}=\log\phi_{10}(\mu_{\rm G}\mid0,I+\Sigma_\ell)
    = -16.81972329. \label{eq:g10_evidence}
\end{align}

\paragraph{Curved Gaussian (CG10).}
This problem retains the standard normal prior of G10 and bends its likelihood with a volume-preserving shear,
\begin{align}
    L_{\rm CG10}(x)&=\phi_{10}\!\left(T_{0.4}^{-1}(x)\mid\mu_{\rm G},\Sigma_\ell\right),
    \label{eq:cg10_likelihood}\\
    T_\beta(z)&=\left(z_1,z_2+\beta\left[(z_1-2)^2-1\right],z_3,\ldots,z_{10}\right).
\end{align}
The shear is centered on the component's latent first-coordinate mean and variance.
It has unit Jacobian and zero mean displacement of the second coordinate under the latent component.
Conditioning $z_{2:10}$ on $z_1$ reduces the reference evidence to one-dimensional quadrature,
\begin{align}
    \log Z_{\rm CG10}
    =\log\E_{z\sim\mathcal{N}(\mu_{\rm G},\Sigma_\ell)}
       \!\left[\phi_{10}(T_{0.4}(z)\mid0,I)\right]
    = -17.89275623. \label{eq:cg10_evidence}
\end{align}

\paragraph{Spike--slab mixture (SS10).}\label{sec:ss10}
We use the spike--slab model of \citet{Albert2025SamplingComparison}, with
\begin{align}
    x&\sim\mathrm{Uniform}\!\left([-4,8]^{10}\right),\notag\\
    L_{\rm SS10}(x)&=\phi_{10}(x\mid\mu_1,0.08I)
                   +\phi_{10}(x\mid\mu_2,0.8I),\label{eq:ss10_likelihood}\\
    \mu_1&=(6,6,0,\ldots,0),\qquad\mu_2=(2.5,2.5,0,\ldots,0).
\end{align}
At equal component Mahalanobis radius, the spike-to-slab volume ratio is $(0.08/0.8)^{10/2}=10^{-5}$.
Writing $s_1=\sqrt{0.08}$, $s_2=\sqrt{0.8}$ and $\Phi$ for the standard normal cumulative distribution function gives the finite-prior reference
\begin{align}
    Z_{\rm SS10}=12^{-10}\sum_{k=1}^{2}\prod_{j=1}^{10}
    \left[\Phi\!\left(\frac{8-\mu_{kj}}{s_k}\right)
         -\Phi\!\left(\frac{-4-\mu_{kj}}{s_k}\right)\right].\label{eq:ss10_evidence}
\end{align}
Thus $\log Z_{\rm SS10}=-24.15593481$, and the reference spike fraction is $0.5000077444$.
The slab-only log evidence is $-24.84909748$.

\begin{figure*}[t]
    \centering
    \includegraphics[width=\textwidth]{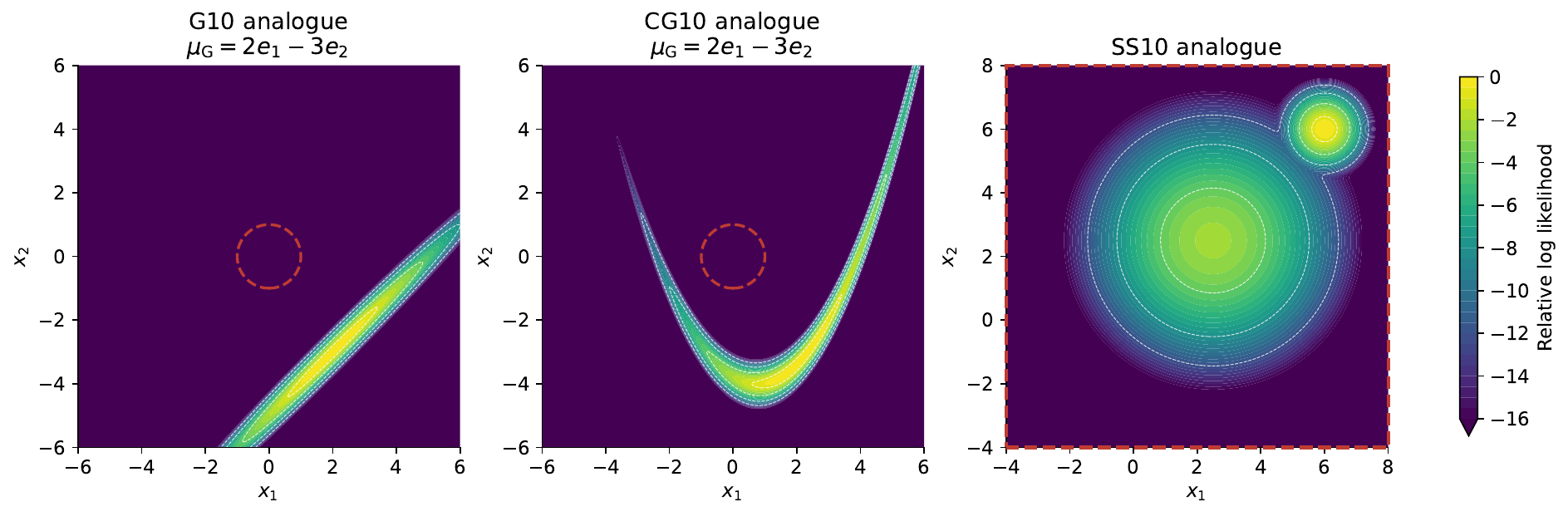}
    \caption{Two-dimensional analogues of G10, CG10, and SS10. Filled contours show relative log likelihood. White lines mark selected levels. Red dashed boundaries show the one-$\sigma$ Gaussian prior circle or the support of the uniform prior. These two-dimensional densities are analogues, not posterior marginals of the ten-dimensional problems. The SS10 geometric volume contrast is much larger in ten dimensions than shown here.}
    \label{fig:evidence_problems}
\end{figure*}

\subsection{Method of assessment}\label{sec:evidence_protocol}

Each experiment uses $R=30$ independent random seeds.
Unless stated otherwise, we use evidence-improving allocation with $\Delta K=30D$ and $d_0=30D$ initial root lineages, where $D$ is the dimension.
The depth condition is the small remaining-evidence condition in Eq.~\eqref{eq:small_remaining_evidence}, with $\tau_Z=\log(1+10^{-3})$.
The goal condition is a classic log-evidence uncertainty below $0.05$, so the run stops at the first completed goal satisfying
\begin{align}
    \widehat\sigma^{\rm exp}_{\log Z}<0.05.\label{eq:benchmark_goal}
\end{align}

For run $r$, let $\widehat\ell_r$ be the mean of its log-evidence shrinkage distribution and $\ell_{\rm ref}=\log Z_{\rm ref}$ the reference log evidence.
We measure accuracy using the root mean squared error,
\begin{align}
    \mathrm{RMSE}=\left[\frac{1}{R}\sum_{r=1}^{R}
    (\widehat\ell_r-\ell_{\rm ref})^2\right]^{1/2}.
\end{align}
The $95\%$ coverage metric is the fraction of runs whose central $95\%$ shrinkage interval contains $\ell_{\rm ref}$.
Each interval is bounded by the $2.5$th and $97.5$th percentiles of that run's log-evidence shrinkage distribution.
A coverage of $95\%$ is ideal, indicating that the intervals contain the reference value at their nominal rate.

\subsection{Evidence accuracy depends on the resolved structure}

\begin{empiricalclaim}\label{claim:resolved_tree}
Conditioning on phantom samples improves evidence accuracy when the classic race tree resolves sufficient likelihood structure.
\end{empiricalclaim}

We use G10 as the initial test of whether phantom conditioning reduces evidence error.
We choose a unimodal yet highly correlated problem, controlling for mode discovery while retaining a non-trivial sampling challenge.
Table~\ref{tab:phantom_evidence_g10_0p05} and Figure~\ref{fig:g10_prefix_accuracy} show that the classic result achieves RMSE $0.05130$, approximately the $0.05$ target, and nominal $95\%$ interval coverage of $93\%$ ($28/30$ runs).
This indicates that the classic method succeeds in both evidence accuracy and uncertainty calibration.
Conditioning on phantom samples then significantly reduces RMSE, with the largest improvement when using all $sD-1$ phantoms.
With $s=10$, this includes all $99$ intermediate states and reduces RMSE to $0.03500$, a drop of $31.8\%$.

\begin{table*}[t]
\centering
\scriptsize
\setlength{\tabcolsep}{4pt}
\caption{G10, $\mu_{\rm G}=2e_1-3e_2$, classic seeds and evidence-improving allocation: $30$ trees, $2048$ paired shrinkage draws per tree, including all $99=sD-1$ phantoms. RMSE SE and pointwise $95\%$ intervals for RMSE minus classic RMSE use $100{,}000$ paired whole-seed bootstrap resamples. Reported SD is mean shrinkage uncertainty. Likelihood calls are $(74.848\pm1.255)\times10^6$ (mean $\pm$ SD), identical for every prefix.}
\label{tab:phantom_evidence_g10_0p05}
\begin{tabular}{lrrrl}
\toprule
Prefix & RMSE $\pm$ SE & Reported SD & 95\% coverage & $\Delta$RMSE 95\% interval \\
\midrule
Classic & 0.05130 $\pm$ 0.00614 & 0.04933 & 93.3\% & --- \\
$1D$ & 0.05799 $\pm$ 0.00782 & 0.02996 & 73.3\% & $[-0.00756,+0.01976]$ \\
$2D$ & 0.05150 $\pm$ 0.00550 & 0.02724 & 63.3\% & $[-0.01196,+0.01203]$ \\
$3D$ & 0.04721 $\pm$ 0.00482 & 0.02532 & 63.3\% & $[-0.01651,+0.00791]$ \\
$4D$ & 0.04279 $\pm$ 0.00476 & 0.02389 & 66.7\% & $[-0.02139,+0.00371]$ \\
$5D$ & 0.04065 $\pm$ 0.00474 & 0.02271 & 70.0\% & $[-0.02377,+0.00156]$ \\
$6D$ & 0.03800 $\pm$ 0.00450 & 0.02175 & 73.3\% & $[-0.02638,-0.00111]$ \\
$7D$ & 0.03767 $\pm$ 0.00422 & 0.02091 & 76.7\% & $[-0.02565,-0.00225]$ \\
$8D$ & 0.03681 $\pm$ 0.00409 & 0.02018 & 76.7\% & $[-0.02648,-0.00320]$ \\
$9D$ & 0.03617 $\pm$ 0.00391 & 0.01952 & 73.3\% & $[-0.02699,-0.00382]$ \\
All ($sD-1$) & 0.03500 $\pm$ 0.00376 & 0.01898 & 70.0\% & $[-0.02805,-0.00509]$ \\
\bottomrule
\end{tabular}
\end{table*}

\begin{figure*}[t]
    \centering
    \includegraphics[width=\textwidth]{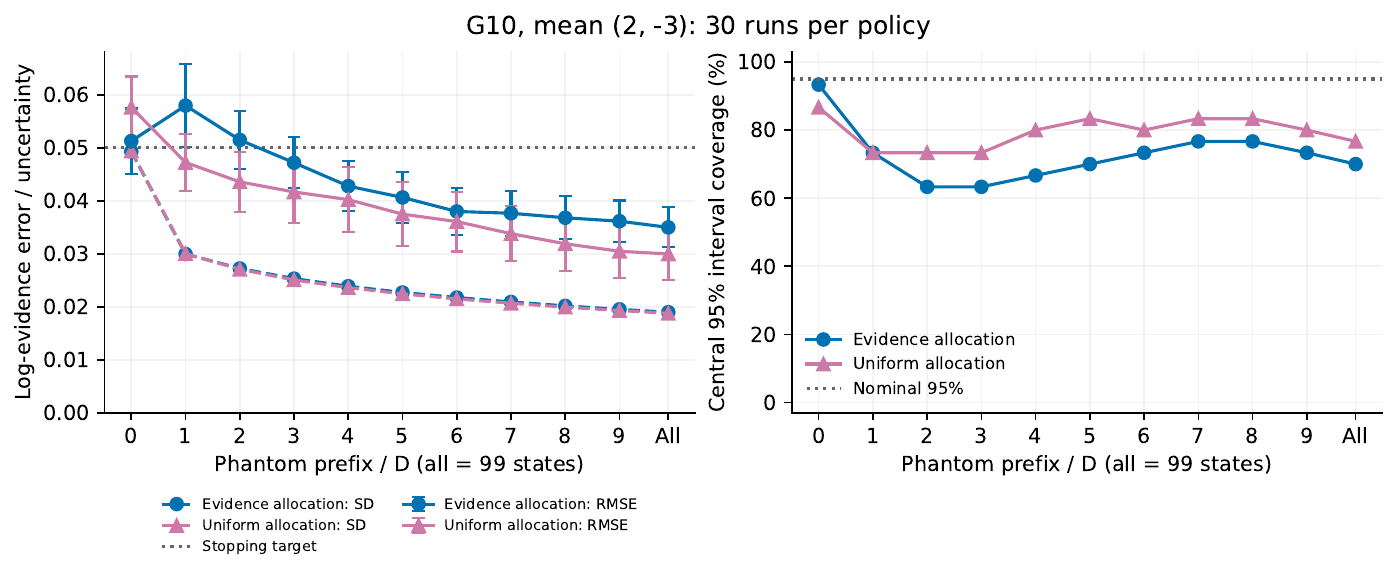}
    \caption{G10 at $\mu_{\rm G}=2e_1-3e_2$: evidence accuracy and calibration for evidence-improving and uniform allocation, with $30$ classic-seed trees per policy. Left: RMSE with bootstrap SE and mean reported shrinkage SD, including all $99=sD-1$ phantoms. Full conditioning reduces RMSE within both policies. The between-policy accuracy differences are unresolved. Right: central $95\%$ interval coverage. Evidence-improving allocation uses $41.4\%$ fewer likelihood evaluations.}
    \label{fig:g10_prefix_accuracy}
\end{figure*}

\begin{empiricalclaim}\label{claim:largest_prefix}
Using all $sD-1$ phantom samples gives the best evidence accuracy within the precision of the experiments.
\end{empiricalclaim}

This is seen directly in the prefix comparison presented in Claim~\ref{claim:resolved_tree}, where the lowest RMSE occurs at $sD-1$.

\begin{empiricalclaim}\label{claim:coverage}
Phantom conditioning produces overconfident evidence uncertainties, with coverage below the nominal $95\%$.
\end{empiricalclaim}

Figure~\ref{fig:g10_prefix_accuracy} shows that the reported uncertainties decrease faster than the RMSE as more phantoms are included.
At the same time, coverage falls from near the ideal $95\%$ to $70\%$ when all phantoms are used.
Thus phantom conditioning improves the evidence estimates while making their reported uncertainties overconfident.

\begin{empiricalclaim}\label{claim:unresolved_structure}
When likelihood structure remains unresolved, conditioning on all phantom samples neither improves nor worsens evidence accuracy within experimental precision.
\end{empiricalclaim}

We use CG10 to make the correlated problem harder by bending its likelihood contours.
Table~\ref{tab:phantom_evidence_cg10_0p05} and Figure~\ref{fig:cg10_prefix_accuracy} show that classic RMSE is $0.07971$, above the $0.05$ target, although the reported uncertainty is near $0.05$.
This indicates that the classic tree has not resolved enough structure to achieve the requested accuracy.
With phantom conditioning, RMSE initially increases at short prefixes, but using all $sD-1$ phantoms gives RMSE $0.08121$, with no resolved difference from the classic result.
Full conditioning therefore preserves evidence accuracy within experimental precision without repairing the unresolved structure.

\begin{table*}[t]
\centering
\scriptsize
\setlength{\tabcolsep}{4pt}
\caption{CG10, $\mu_{\rm G}=2e_1-3e_2$, classic seeds and evidence-improving allocation: $30$ trees, $2048$ paired shrinkage draws per tree, including all $99=sD-1$ phantoms. RMSE SE and pointwise $95\%$ intervals for RMSE minus classic RMSE use $100{,}000$ paired whole-seed bootstrap resamples. Reported SD is mean shrinkage uncertainty. Likelihood calls are $(91.171\pm2.620)\times10^6$ (mean $\pm$ SD), identical for every prefix.}
\label{tab:phantom_evidence_cg10_0p05}
\begin{tabular}{lrrrl}
\toprule
Prefix & RMSE $\pm$ SE & Reported SD & 95\% coverage & $\Delta$RMSE 95\% interval \\
\midrule
Classic & 0.07971 $\pm$ 0.00661 & 0.04965 & 66.7\% & --- \\
$1D$ & 0.10276 $\pm$ 0.00895 & 0.02975 & 26.7\% & $[+0.00469,+0.03931]$ \\
$2D$ & 0.09800 $\pm$ 0.00923 & 0.02737 & 26.7\% & $[-0.00122,+0.03572]$ \\
$3D$ & 0.09490 $\pm$ 0.00874 & 0.02584 & 26.7\% & $[-0.00347,+0.03195]$ \\
$4D$ & 0.09209 $\pm$ 0.00840 & 0.02467 & 30.0\% & $[-0.00577,+0.02860]$ \\
$5D$ & 0.08808 $\pm$ 0.00822 & 0.02377 & 30.0\% & $[-0.00944,+0.02444]$ \\
$6D$ & 0.08644 $\pm$ 0.00794 & 0.02299 & 30.0\% & $[-0.01061,+0.02249]$ \\
$7D$ & 0.08454 $\pm$ 0.00760 & 0.02231 & 30.0\% & $[-0.01191,+0.02011]$ \\
$8D$ & 0.08310 $\pm$ 0.00742 & 0.02172 & 26.7\% & $[-0.01291,+0.01826]$ \\
$9D$ & 0.08213 $\pm$ 0.00723 & 0.02121 & 23.3\% & $[-0.01337,+0.01672]$ \\
All ($sD-1$) & 0.08121 $\pm$ 0.00697 & 0.02078 & 20.0\% & $[-0.01383,+0.01537]$ \\
\bottomrule
\end{tabular}
\end{table*}

\begin{figure*}[t]
    \centering
    \includegraphics[width=\textwidth]{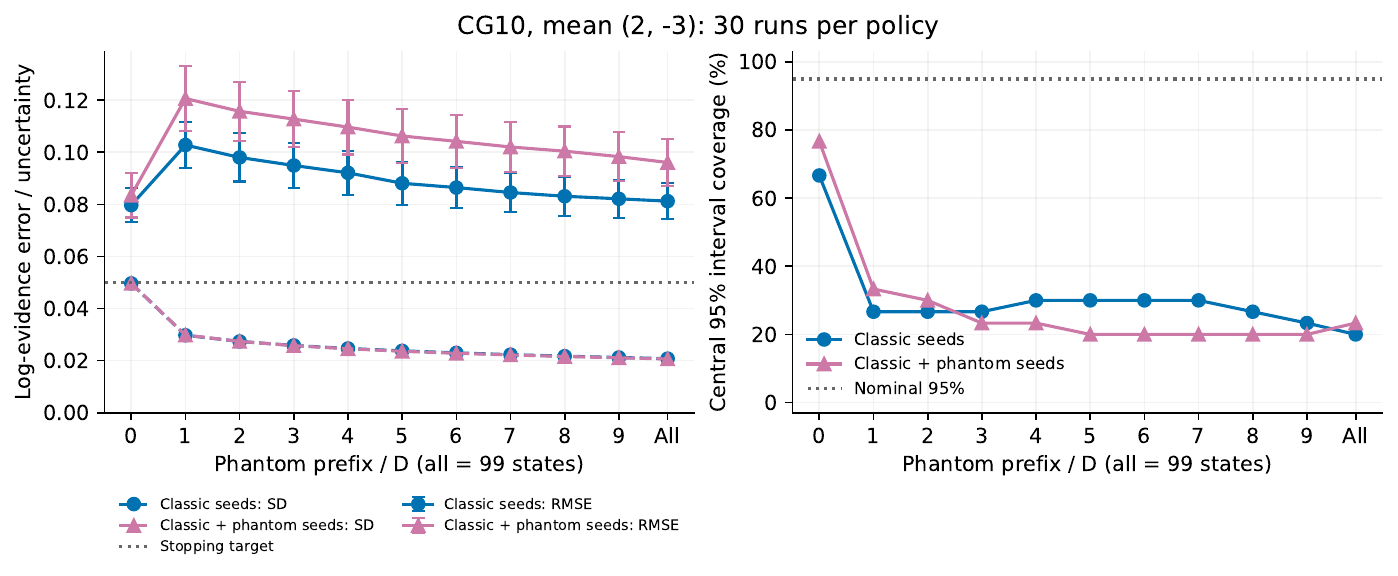}
    \caption{CG10 at $\mu_{\rm G}=2e_1-3e_2$: classic-only versus classic-plus-phantom seed populations, with $30$ trees per policy and evidence-improving allocation. Left: RMSE with bootstrap SE and mean reported shrinkage SD, including all $99$ phantoms. Phantom seeding has higher point RMSE, but the between-policy differences are unresolved. Right: central $95\%$ coverage. Both seed policies substantially understate evidence error.}
    \label{fig:cg10_prefix_accuracy}
\end{figure*}

SS10 supports this claim in the extreme case where an entire mode is missed.
Figure~\ref{fig:ss10_classic_seed_mass} shows the variation in recovered spike mass across runs, with many runs assigning almost no mass to a mode that carries half the posterior mass.
Across $30$ runs, classic and full-phantom evidence RMSE are $0.75004$ and $0.73695$, respectively, showing no resolved improvement or worsening despite the missed structure (Table~\ref{tab:phantom_evidence_ss10_0p05}, Figure~\ref{fig:ss10_prefix_accuracy}).

\begin{table*}[t]
\centering
\scriptsize
\setlength{\tabcolsep}{4pt}
\caption{SS10, classic seeds and evidence-improving allocation, $30$ trees and $2048$ paired shrinkage draws per tree. All $99=sD-1$ phantoms are included at the final endpoint. RMSE SE and pointwise $95\%$ intervals for RMSE minus classic RMSE use $100{,}000$ paired whole-seed bootstrap resamples. Reported SD is the mean within-tree shrinkage uncertainty. Mean likelihood calls are $(32.196\pm17.356)\times10^6$ (mean $\pm$ SD). Costs and classic posteriors are fixed across prefixes.}
\label{tab:phantom_evidence_ss10_0p05}
\begin{tabular}{lrrrl}
\toprule
Prefix & RMSE $\pm$ SE & Reported SD & 95\% coverage & $\Delta$RMSE 95\% interval \\
\midrule
Classic & 0.75004 $\pm$ 0.06465 & 0.04917 & 10.0\% & --- \\
$1D$ & 0.73644 $\pm$ 0.05957 & 0.02880 & 6.7\% & $[-0.03165,+0.00396]$ \\
$2D$ & 0.73766 $\pm$ 0.05920 & 0.02378 & 0.0\% & $[-0.03362,+0.00651]$ \\
$3D$ & 0.73564 $\pm$ 0.05947 & 0.02086 & 0.0\% & $[-0.03644,+0.00462]$ \\
$4D$ & 0.73734 $\pm$ 0.05982 & 0.01890 & 0.0\% & $[-0.03317,+0.00559]$ \\
$5D$ & 0.73667 $\pm$ 0.05978 & 0.01755 & 0.0\% & $[-0.03321,+0.00467]$ \\
$6D$ & 0.73640 $\pm$ 0.05980 & 0.01649 & 0.0\% & $[-0.03355,+0.00464]$ \\
$7D$ & 0.73634 $\pm$ 0.05975 & 0.01563 & 0.0\% & $[-0.03325,+0.00457]$ \\
$8D$ & 0.73678 $\pm$ 0.05961 & 0.01492 & 0.0\% & $[-0.03259,+0.00500]$ \\
$9D$ & 0.73700 $\pm$ 0.05955 & 0.01431 & 0.0\% & $[-0.03226,+0.00497]$ \\
All ($sD-1$) & 0.73695 $\pm$ 0.05958 & 0.01386 & 0.0\% & $[-0.03176,+0.00455]$ \\
\bottomrule
\end{tabular}
\end{table*}

\begin{figure*}[t]
    \centering
    \includegraphics[width=\textwidth]{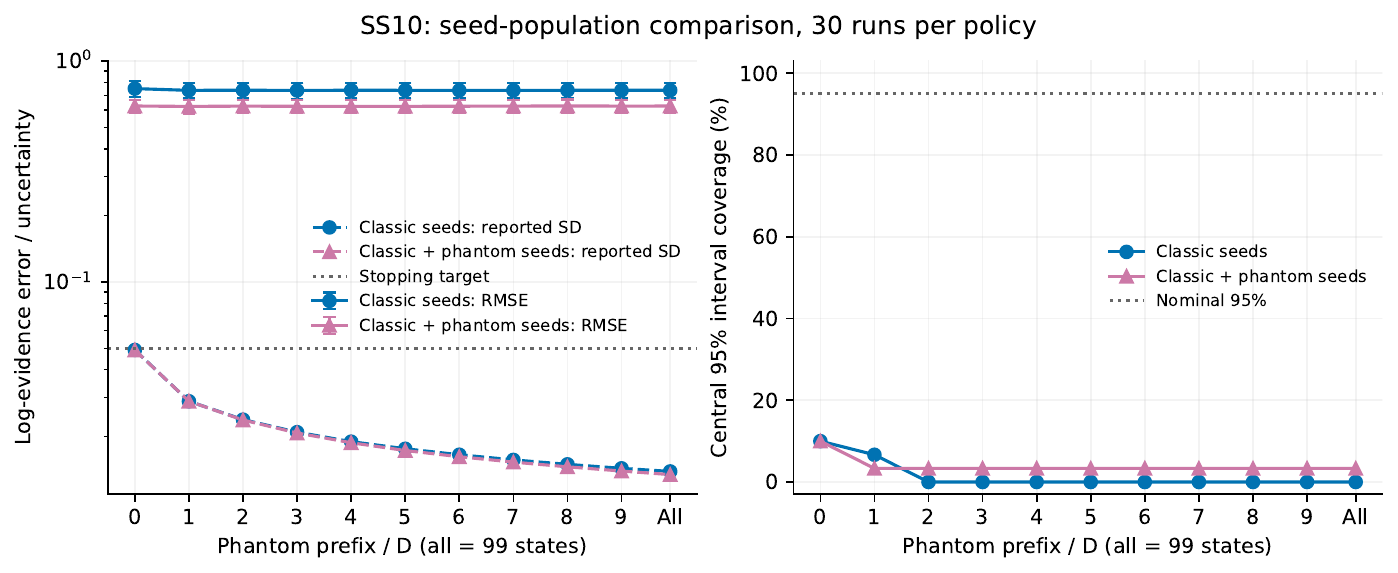}
    \caption{SS10 evidence accuracy and calibration for classic-only seeds and classic-plus-phantom seeds, with $30$ trees per policy and all $99=sD-1$ phantoms at the final endpoint. Left: RMSE with bootstrap SE and mean reported shrinkage SD. Right: central $95\%$ interval coverage. Each policy's prefix sweep holds its classic tree and posterior fixed. Phantom seeding has lower point RMSE, but the paired between-policy difference is unresolved. Both policies remain severely inaccurate and overconfident.}
    \label{fig:ss10_prefix_accuracy}
\end{figure*}

\subsection{Phantom seeds and dependence between chains}\label{sec:phantom_seed_dependence}

We explore hypotheses about whether adding phantom seeds improves mode recovery or increases inter-cluster correlation and worsens evidence accuracy.
By phantom seeding, we mean appending phantom samples to the classic sample population and allowing them to be selected as seeds under the same eligibility rule, with all other sampling settings unchanged.

\begin{empiricalclaim}\label{claim:phantom_seeds}
Including phantoms in the seed population provides no demonstrated accuracy benefit and may worsen estimates through increased inter-cluster correlation.
\end{empiricalclaim}

We rerun CG10 with phantom seeding.
Figure~\ref{fig:cg10_prefix_accuracy} shows higher RMSE with phantom seeds, increasing from $0.08121$ to $0.09606$ at full conditioning, although the difference is not statistically resolved.
Phantom samples provide stationary seeds, but reusing them to start new chains necessarily introduces correlation between clusters.
Classic samples form the most independent seed population available, providing the stationary and approximately independent seeds assumed by phantom conditioning.

\begin{empiricalclaim}\label{claim:phantom_mode_recovery}
Including phantom seeds does not reliably recover missing or mismeasured mode mass.
\end{empiricalclaim}

We apply the same seed-selection comparison to SS10, where a narrow mode carries half the posterior mass, while the spike-to-slab volume ratio in prior space is $10^{-5}$.
We measure its recovered mass using the component responsibilities,
\begin{align}
    \widehat m_1=\sum_i w_i^{\rm classic}
    \frac{L_1(x_i)}{L_1(x_i)+L_2(x_i)}.
\end{align}
Figure~\ref{fig:ss10_classic_seed_mass} shows essentially unchanged mode-mass RMSE: $0.39623$ with classic seeds and $0.39752$ with phantom seeds.
Neither seed population reliably recovers the mode mass.
Full-phantom evidence RMSE is lower with phantom seeds, $0.62529$ versus $0.73695$ (Figure~\ref{fig:ss10_prefix_accuracy}), but this difference is not statistically resolved and does not translate into reliable recovery of mode mass.

\begin{figure*}[t]
    \centering
    \includegraphics[width=\textwidth]{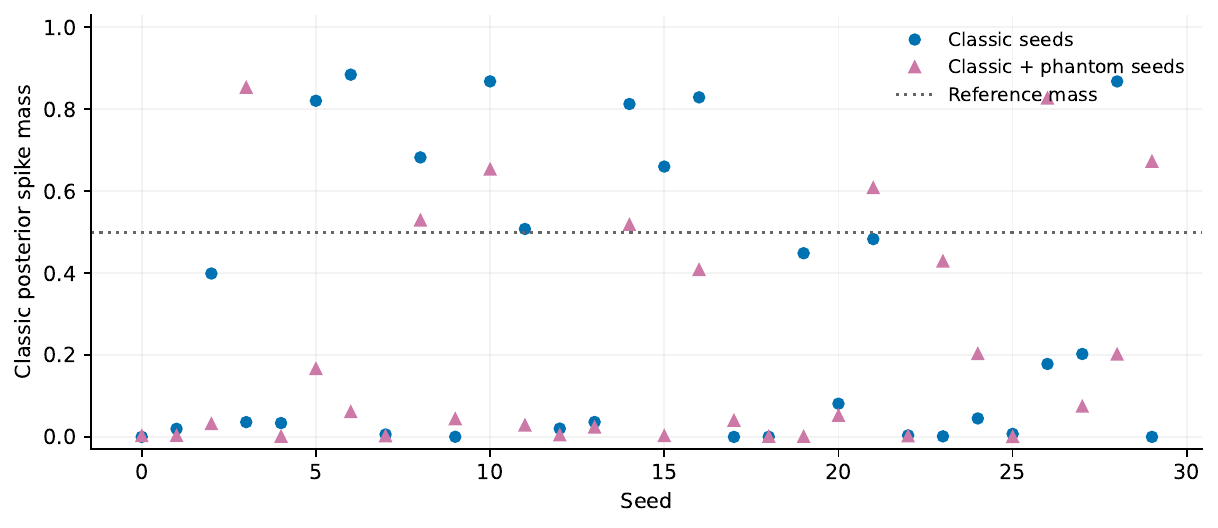}
    \caption{Classic posterior spike responsibility mass for all $30$ SS10 seeds under each seed-selection policy. The dotted line is the finite-prior reference mass. Both policies show severe mass errors, with no resolved difference in mode-mass RMSE. Within each policy, every evidence-prefix reduction shares these exact classic posterior masses.}
    \label{fig:ss10_classic_seed_mass}
\end{figure*}

\subsection{Mode death comparison with dynesty and polychord}\label{sec:mode_death_comparison}

We compare JAXNS with other packages on the same SS10 problem to assess how they handle mode death.
We run \texttt{dynesty} 3.1.0 and PolyChordLite 1.22.3 using similar likelihood-evaluation budgets, with $30$ random seeds for each sampler, and measure mode-mass RMSE.
\texttt{dynesty} uses $300$ initial and batch live points, multiple-ellipsoid bounds, $30$-step random walks, and evidence allocation.
Its likelihood budgets are matched to the corresponding classic-seed JAXNS runs, with completed batches exceeding those budgets by $0.188\%$ on average.

PolyChord uses clustering, $50$ whitened slice repeats, and a $10^{-3}$ remaining-evidence threshold.
Independent cost-only pilot runs set its live populations to $2682$--$10063$ to match the individual likelihood budgets.
Its realised mean likelihood work is $28.5\%$ greater than the JAXNS baseline, so this match is approximate.

Table~\ref{tab:ss10_sampler_comparison} shows that \texttt{dynesty} performs worse than JAXNS, while PolyChord performs better.
PolyChord explicitly tracks cluster volumes when allocating new live points \citep{2015MNRAS.453.4384H}.
We conjecture that its better recovery comes from identifying modes early and preserving their relative prior volumes, suggesting that explicit mode tracking is likely required to mitigate mode death.

\begin{table*}[t]
\centering
\scriptsize
\setlength{\tabcolsep}{4pt}
\caption{SS10 mode-mass recovery and mean likelihood work, $30$ seeds per row. The reference spike fraction is $0.5000077444$. External sampler budgets were approximately matched to the classic-seed JAXNS baseline, with PolyChord using $28.5\%$ more likelihood evaluations.}
\label{tab:ss10_sampler_comparison}
\begin{tabular}{llrrr}
\toprule
Sampler / seeds & Stop & Calls ($10^6$) & Mean mass & Mass RMSE\\
\midrule
JAXNS / classic & $0.05$ & 32.196 & 0.29798 & 0.39623\\
JAXNS / phantom & $0.05$ & 28.742 & 0.21483 & 0.39752\\
\texttt{dynesty} & matched work & 32.257 & 0.01693 & 0.48348\\
PolyChord & approximate match & 41.384 & 0.36297 & 0.27308\\
\bottomrule
\end{tabular}
\end{table*}

\subsection{Allocation efficiency and dimensionality}

\begin{empiricalclaim}\label{claim:evidence_allocation}
Evidence-improving allocation approximately halves the likelihood evaluations required without a detectable loss of evidence accuracy.
\end{empiricalclaim}

Uniform allocation recovers the allocation of static nested sampling, which we use as the gold standard for this comparison.
We compare evidence-improving and uniform allocation on G10 at the same $0.05$ classic uncertainty target.
Evidence-improving allocation requires $(74.848\pm1.255)\times10^6$ likelihood evaluations, compared with $(127.723\pm2.852)\times10^6$ for uniform allocation.
This saves $41.4\%$ of the likelihood evaluations.
Figure~\ref{fig:g10_prefix_accuracy} shows comparable RMSE under both policies for classic and phantom-conditioned estimates, so the reduction in work comes without a detectable loss of accuracy.

\begin{empiricalclaim}\label{claim:dimension}
Phantom conditioning combined with evidence-improving allocation remains effective as dimension increases, with larger evidence-accuracy improvements at higher dimensions.
\end{empiricalclaim}

We test this claim using the Gaussian model at $D=10,20,40$ (G10, G20, and G40) with evidence-improving allocation.
Initial root lineages and allocation increments scale as $30D$, while execution width and the number of slice transitions per chain scale as $10D$.
Table~\ref{tab:dimension_comparison} and Figure~\ref{fig:dimension_comparison} compare classic and full-phantom RMSE together with the likelihood cost.
Classic RMSE remains near the $0.05$ target in all three dimensions, while full-phantom RMSE falls to $0.03500$, $0.02085$, and $0.01461$, respectively.
The corresponding improvements are $31.8\%$, $54.2\%$, and $72.1\%$.
Thus the accuracy benefit grows with dimension, although the likelihood cost also increases.
We conjecture that the improvement grows with dimension because the effective participating-cluster count $C_g^{\rm Kish}$ (Eq.~\ref{eq:kish_cluster_count}) increases, supplying more independent phantom Monte Carlo observations.

\begin{table*}[t]
\centering
\scriptsize
\setlength{\tabcolsep}{3pt}
\caption{Completed Gaussian dimension cohorts with classic seeds and evidence-improving allocation, $30$ runs per dimension. RMSE uncertainties are bootstrap SEs. Likelihood counts are mean $\pm$ across-run SD. Full retention uses $99$ phantoms in $10$ dimensions, $199$ in $20$, and $399$ in $40$.}
\label{tab:dimension_comparison}
\begin{tabular}{rrrr}
\toprule
$D$ & Classic RMSE & Full RMSE & Calls ($10^6$) \\
\midrule
10 & $0.0513\pm0.0061$ & $0.0350\pm0.0038$ & $74.85\pm1.26$ \\
20 & $0.0455\pm0.0066$ & $0.0208\pm0.0018$ & $488.08\pm5.25$ \\
40 & $0.0524\pm0.0051$ & $0.0146\pm0.0020$ & $3486.33\pm53.40$ \\
\bottomrule
\end{tabular}
\end{table*}

\begin{figure*}[t]
    \centering
    \includegraphics[width=\textwidth]{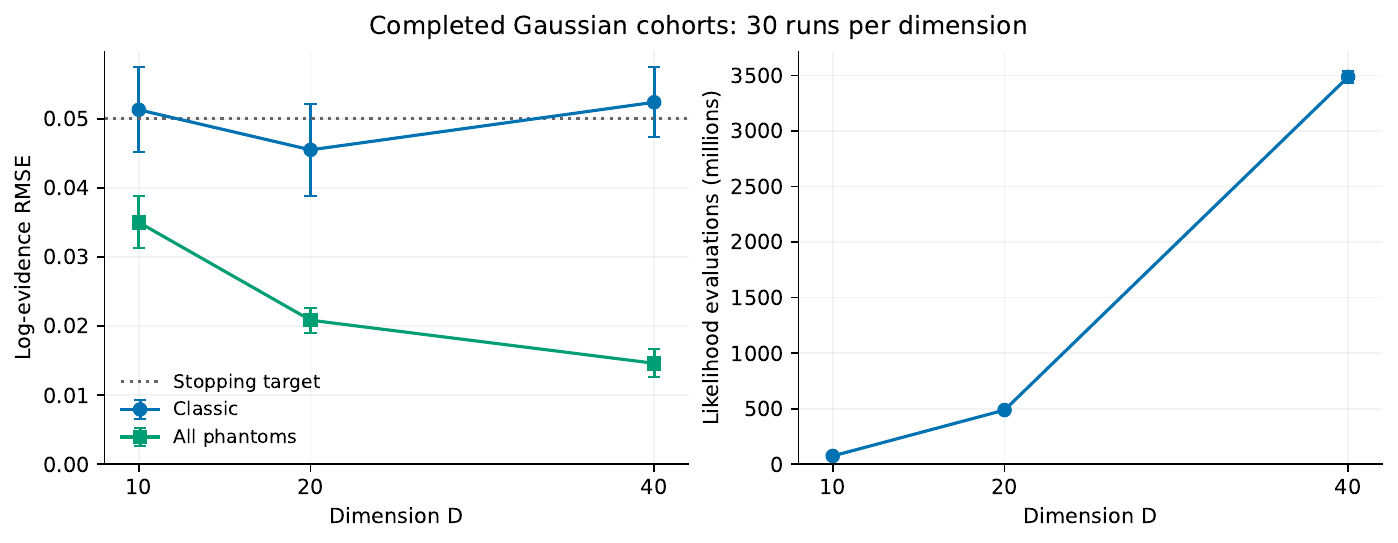}
    \caption{Completed Gaussian dimension cohorts at $\mu_{\rm G}=2e_1-3e_2$, with classic seeds and evidence-improving allocation. Left: classic and full-retention evidence RMSE, with bootstrap SE. Full retention uses $99$, $199$, and $399$ phantoms in $10$, $20$, and $40$ dimensions, respectively. Right: mean likelihood evaluations with across-run SD. Each point uses all $30$ seeds.}
    \label{fig:dimension_comparison}
\end{figure*}

\subsection{Posterior allocation after evidence stopping}

\begin{empiricalclaim}\label{claim:posterior_allocation}
Posterior-improving allocation cheaply increases effective sample size within the resolved structure, and does not recover missing structure.
\end{empiricalclaim}

We resume each of the $30$ CG10 seed experiments using posterior-improving allocation, with a goal condition of doubling the classic Kish ESS,
\begin{align}
    \mathrm{ESS}_{\rm Kish}=\frac{\left(\sum_i w_i^{\rm classic}\right)^2}{\sum_i(w_i^{\rm classic})^2},
    \qquad \mathrm{ESS}_{\rm final}\ge2\,\mathrm{ESS}_{\rm initial}.
\end{align}
Mean ESS increases from $28{,}443$ to $57{,}938$ for only $17.5\%$ additional likelihood evaluations.
Both classic and full-phantom evidence RMSE remain near $0.08$, above the $0.05$ target.
Posterior-improving allocation therefore cheaply refines the resolved structure without repairing the evidence error caused by unresolved structure.

\section{Conclusions}\label{sec:conclusion}

This paper reformulates NS shrinkage as racing lineages, which decouples the shrinkage arithmetic from the sample generation process.
We introduce a Bayesian model for treating plateaus.
We show that plateaus censor the race rank: we know that the race finishes, but not when.
Building on this model, we condition on phantom samples by using their Monte Carlo estimates of contour exceedance probabilities as observations.
The resulting method uses information that Markov-chain constrained samplers already generate, while preserving the order-statistic law of the classic sequence.
This framework is implemented in JAXNS v3, which is presented along with its implementation design.

Experiments support phantom conditioning as a way to improve evidence accuracy when the classic race tree resolves sufficient likelihood structure.
Using all phantom samples gives the best accuracy within experimental precision, with larger improvements at higher dimensions.
We conjecture that this stronger improvement arises because the effective number of independent phantom Monte Carlo observations increases with dimension.
When relevant structure is unresolved, full phantom conditioning leaves evidence accuracy unchanged within experimental precision.
Phantom-conditioned uncertainties are overconfident, and unresolved structure can also make classic uncertainties overconfident.

Including phantoms in the seed population provides no demonstrated accuracy benefit and may worsen evidence estimates through increased inter-cluster correlation.
Phantom seeds do not reliably recover mode mass.
We conjecture that identifying modes early and tracking their relative prior volumes is likely required to mitigate mode death.

Evidence-improving allocation approximately halves the likelihood evaluations required relative to uniform allocation without a detectable loss of evidence accuracy.
Combined with full phantom conditioning, its accuracy benefit persists as dimension increases.
Posterior-improving allocation cheaply increases the effective sample size within the resolved structure, and does not recover missing structure.

\backmatter

\bmhead{Supplementary information}

Code and data required to reproduce the experiments is provided in the JAXNS repository at \url{https://www.github.com/joshuaalbert/jaxns}.

\bmhead{Acknowledgments}

JA would like to thank John Skilling for discussions around nested sampling, as well as Will Handley, Andrew Fowlie, and Johannes Buchner for pointing out flaws in earlier attempts to use phantom samples.

\begin{appendices}

\section{Alternative Allocation Targets}\label{sec:allocation_targets}
The allocation targets define how likelihood evaluation work is dedicated per outer loop.
The two utility-based gaps are,
\begin{align}
    G_g^{Z,k} &= \left\lceil\Delta K \bar{U}^Z_g\right\rceil,\\
    G_g^{P,k} &= \left\lceil\Delta K \bar{U}^P_g\right\rceil,
\end{align}
where $\Delta K > 0$, and $\bar{U}^Z_g$ and $\bar{U}^P_g$ are defined below.

Evidence improving allocation seeks to allocate new samples to contours by maximising the expected reduction in evidence variance.
A key observation is that a new child generated from parent block $g$ changes the parent out-degree as $d_g\to d_g+1$, Eq.~\ref{eq:K_recursion}.
The probability of that sample reaching $L_h$ is given by,
\begin{align}
    T^Z_{g h}
    \triangleq
    \Prob_{x\sim\pi_{L_g}}(L(x)>L_h)
    = \frac{X_h}{X_g} \mathbf{1}\{h>g\}.
    \label{eq:evidence_reach_probability}
\end{align}
That is, the probability that we add one additional active lineage count to blocks $g+1$ to $l$ is $T^Z_{g h}$.

We will define utility function as the expected reduction in variance.
To simplify calculation, we will use a first-order delta approximation to the log-evidence variance in terms of the log-shrinkage variance.
Structurally, we have $\mathrm{Var}_\theta[f] \approx \sum_\theta (\partial_\theta f)^2\, \mathrm{Var}[\theta]$.
Applying this to $\log Z$ given by Eq.~\ref{eq:block_quadrature} with block shrinkage $p_{>h} \sim \Beta(\alpha_h, \beta_h)$, and recalling that $X_h = p_{>h} X_{h-1} = \prod_j^h p_{>j}$ we have,
\begin{align}
    B_h \triangleq \frac{\partial \log Z}{\partial \log p_{>h}}
    =&
    \frac{
    L_h X_h
    -\sum_{j > h} L_j( X_{j-1}- X_j)
    }{ Z},\label{eq:evidence_sensitivity}\\
    \mathrm{Var}[\log p_{>h}] =& \psi_1(\alpha_h) - \psi_1(\alpha_h + \beta_h),
\end{align}
where $\psi_1$ is the tri-gamma function.
The first term in Eq.~\ref{eq:evidence_sensitivity} is the contribution from changing the block $h$, and the second term is the accumulated downstream effect because $p_{>h}$ appears in every later term.
Now the effect of adding one extra active lineage is that $\alpha_h \to \alpha_h + 1$, so the reduction in variance of log-evidence under this delta approximation is given by,
\begin{align}
    \Delta_h \mathrm{Var}[\log Z] \approx B_h^2 \left(\frac{1}{\alpha_h^2} - \frac{1}{(\alpha_h + \beta_h)^2}\right),\label{eq:delta_var_logZ}
\end{align}
where we used the relationship $\psi_1(z) - \psi_1(z + 1) = 1/z^2$.
Now, the utility is the expectation of Eq.~\ref{eq:delta_var_logZ} over the transition probability, Eq.~\ref{eq:evidence_reach_probability},
\begin{align}
    U^Z_g = \frac{1}{X_g}\sum_{h > g} X_h B_h^2 \left(\frac{1}{\alpha_h^2} - \frac{1}{(\alpha_h + \beta_h)^2}\right).\label{eq:evidence_improving_utility}
\end{align}
This can be computed efficiently in two cumulative reductions.

Posterior improving allocation seeks to increase the effective sample size, by allocating samples low-information shells defined by the half-open likelihood intervals, $(L_{h-1}, L_h]$.
A key observation is that a uniform draw within contour $L_g$ will land in shell $h$ with probability,
\begin{align}
    T^P_{g h}
    \triangleq
    \Prob_{x\sim\,\pi_{L_{g}}}(L_{h-1} < L(x) \le L_h)
    \approx \frac{ X_{h-1} -  X_{h}}{ X_g} \,\mathbf{1}\{h>g\}.
\end{align}
Recalling the posterior shell masses, Eq.~\ref{eq:post_weight_cont}, let us use Kish's effective sample size as a proxy for effective sample size,
\begin{align}
    \mathrm{ESS}^{\rm Kish} =& \frac{W^2}{Q}\\
    W\triangleq& \sum_{j=1}^{G} w_j,\\
    Q\triangleq& \sum_{j=1}^{G} w_j^2.
\end{align}
If a new sample lands in shell $h$, it splits the shell into two sub-shells.
In this case $W$ is unchanged under additivity, and only $Q$ changes.
Writing $u\sim\,\mathcal{U}[0, 1]$ for the fraction of the shell mass assigned to the lower sub-shell, the old contribution $w_h^2$ to $Q$ is replaced by $u^2 w_h^2+(1-u)^2 w_h^2$.
Thus,
\begin{align}
    \mathrm{ESS}^{\rm Kish}_{h}(u)
    =\frac{W^2}{Q - 2u(1-u)w_h^2}.
\end{align}
If we marginalise over $u$ the change of $\mathrm{ESS}^{\rm Kish}$ due to splitting in $h$ we get,
\begin{align}
    \Delta\mathrm{ESS}^{\rm Kish}_h
    \triangleq&
    \int_0^1
    \left(\mathrm{ESS}^{\rm Kish}_h(u)-\mathrm{ESS}^{\rm Kish}\right)\,\dd u \\
    &=
    W^2\left[
    \frac{
    \tan^{-1}\left(\sqrt{w_h^2/(2Q-w_h^2)}\right)
    }
    {\sqrt{(Q-w_h^2/2)(w_h^2/2)}}
    -\frac{1}{Q}
    \right],
    \label{eq:posterior_shell_value}
\end{align}
with the continuous limiting value $\Delta\mathrm{ESS}^{\rm Kish}_h \to 0$ as $w_h \to 0$.
A simple conservative approximation is obtained by replacing the denominator with its expectation,
\begin{align}
    \Delta\mathrm{ESS}^{\rm Kish}_h
    \approx
    \frac{W^2}{Q - w_h^2/3} - \frac{W^2}{Q}.
\end{align}
This is conservative because $1/x$ is convex so $\E[1/x] \ge 1/\E[x]$.
Now the utility is found by marginalising over the shell transition probabilities that a child from parent block $g$ lands in that shell $h$ gives the posterior-improving utility,
\begin{align}
    U^P_g
    =
    \frac{1}{X_g}\sum_{h > g}
     (X_{h-1}- X_h)\,\Delta\mathrm{ESS}^{\rm Kish}_h .
    \label{eq:posterior_improving_utility}
\end{align}

\end{appendices}

\bibliography{cite}

\end{document}